\def\year{2022}\relax
\documentclass[letterpaper]{article} % DO NOT CHANGE THIS
\usepackage{aaai22}
\usepackage{times}  % DO NOT CHANGE THIS
\usepackage{helvet}  % DO NOT CHANGE THIS
\usepackage{courier}  % DO NOT CHANGE THIS
\usepackage[hyphens]{url}  % DO NOT CHANGE THIS
\usepackage{graphicx} % DO NOT CHANGE THIS
\usepackage{hyperref}
\usepackage{natbib}  % DO NOT CHANGE THIS AND DO NOT ADD ANY OPTIONS TO IT
\usepackage{caption} % DO NOT CHANGE THIS AND DO NOT ADD ANY OPTIONS TO IT
\DeclareCaptionStyle{ruled}{labelfont=normalfont,labelsep=colon,strut=off} % DO NOT CHANGE THIS
\usepackage{algorithm}
\usepackage{algorithmic}
\usepackage{newfloat}
\usepackage{listings}
\floatstyle{ruled}
\newfloat{listing}{tb}{lst}{}
\floatname{listing}{Listing}
\usepackage{amsmath, amssymb, mathtools}
\usepackage{amsthm}
\usepackage{dsfont}
\usepackage{caption}
\usepackage{xcolor}
\usepackage{booktabs}

\title{Risk-Neutral Market Simulation}
 
\author{
    Magnus Wiese\textsuperscript{\rm 1, \rm 2},
    Phillip Murray\textsuperscript{\rm 1, \rm 3}
}
\affiliations{
    \textsuperscript{\rm 1}J.P. Morgan, 25 Bank Street, London, United Kingdom\thanks{Opinions expressed in this paper are those of the authors, and do not necessarily reflect the view of J.P. Morgan.}
    \\
    \textsuperscript{\rm 2}University of Kaiserlautern, Gottlieb-Daimler Stra{\ss}e 48, Kaiserslautern, Germany \\
    \textsuperscript{\rm 3}Imperial College London, 180 Queen's Gate, London, United Kingdom \\
}

\usepackage{xcolor}

\DeclareRobustCommand{\E}[0]{\mathbb{E}}

\DeclareRobustCommand{\N}[0]{\mathbb{N}}

\DeclareRobustCommand{\P}[0]{\mathbb{P}}
\DeclareRobustCommand{\Q}[0]{\mathbb{Q}}
\DeclareRobustCommand{\R}[0]{\mathbb{R}}

\DeclareRobustCommand{\X}[0]{\mathbb{X}}

\newcommand{\KL}{D_{\textrm{KL}}}

\DeclareMathOperator{\Image}{Im}

\newcommand{\pdens}{p}

\newcommand{\qdens}{q}

\renewcommand{\X}[0]{X}

\newcommand{\x}{x}
\newcommand{\y}{y}

\newtheorem{proposition}{Proposition}
\theoremstyle{remark}

\begin{document}

\maketitle

\begin{abstract}
    We develop a risk-neutral spot and equity option market simulator for a single underlying, under which the joint market process is a martingale. We leverage an efficient low-dimensional representation of the market which preserves no static arbitrage, and employ neural spline flows to simulate samples which are free from conditional drifts and are highly realistic in the sense that among all possible risk-neutral simulators, the obtained risk-neutral simulator is the closest to the historical data with respect to the Kullback-Leibler divergence. Numerical experiments demonstrate the effectiveness and highlight both drift removal and fidelity of the calibrated simulator. 
\end{abstract}

\section{Introduction}
There has been growing interest in the application of reinforcement learning methods to financial markets. In many domains, a significant difficulty is the the availability of sufficient historical data to train reinforcement learning agents on -- for example taking daily data for an equity underlying over ten years amounts to only a few thousand samples, making any algorithm prone to overfitting and lacking robustness. To address this data scarcity challenge, there has been much work in the area of using generative machine learning models to simulate realistic samples from the same distribution as the historical market data. To name a few: \cite{arribas2020sigsdes, horvath2020, cohen2021arbitrage, cohen2022estimating, cuchiero2020generative, de2021tackling, gierjatowicz2020robust, hao2020, ni2021sig, wiese2019deep, wiese2020quant}. 

In this paper, we focus on the simulation of equity option markets of a single underlying. Through an effective parametrization of the market in terms of \textit{discrete local volatilities} \cite{DLV}, we are able to guarantee the absence of static arbitrage from the market simulator. However, a major difficulty in the simulation of equity option markets is the presence of \textit{statistical arbitrage} -- that is, trading strategies which, starting from an empty initial portfolio, have positive expectation. This poses a particular challenge to the robustness of reinforcement agents trained on simulated data, for the purposes of hedging and risk management of a derivatives portfolio. In the presence of statistical arbitrage, an agent's proposed action will be not only a hedge, but also an additional component that is independent of the portolio it is hedging and purely seeking profitable opportunities in the market, which are particularly sensitive to errors in estimating the drift (conditional expectation) of the tradeable instruments.

It is therefore desirable to build models which are free from such statistical arbitrage opportunities. In the absence of trading constraints, such as transaction costs, removing statistical arbitrage is equivalent to simulating from a \textit{risk-neutral measure} -- a distribution under which the (discounted) prices process of all traded instruments are martingales: $\E[X_{t+1}|\mathcal{F}_t ]=X_t$. Simulating from such a martingale measure is a long-standing challenge in quantitative finance, with the classic approach being to specify a suitable parametric model for the underlying under the risk-neutral measure and calibrate parameters to historical data. These models are, however, often motivated more by tractability than expressiveness, and are often limited to single underlyings, not to the joint distributions of an underlying and options written on it. 

Therefore, we take a different approach in this paper and demonstrate how machine learning methods can be used for direct simulation under risk-neutral measures. We build on previous works on market simulation via flow-based models, first introduced in \cite{wiese2021multiasset}, and on methods to ``remove the drift'' in a market by solving a utility optimization problem \cite{buehler2022deep}, in order to create a new data-driven method for direct simulation from a risk-neutral measure, given historical samples from the physical measure. 

The use case and advantages of this approach includes for providing training data for a \textit{Deep Hedging} algorithm \cite{DH} to construct hedging strategies for portfolios of derivatives, by trading in hedging intruments consisting of an underlying and vanilla options. By training the algorithm with a risk-neutral simulator, we can remove the statistical arbitrage component, leaving a hedge that is robust against estimation errors of the drift \cite{buehler2022deep}. 

\begin{figure*}
    \centering
    \includegraphics[width=\textwidth]{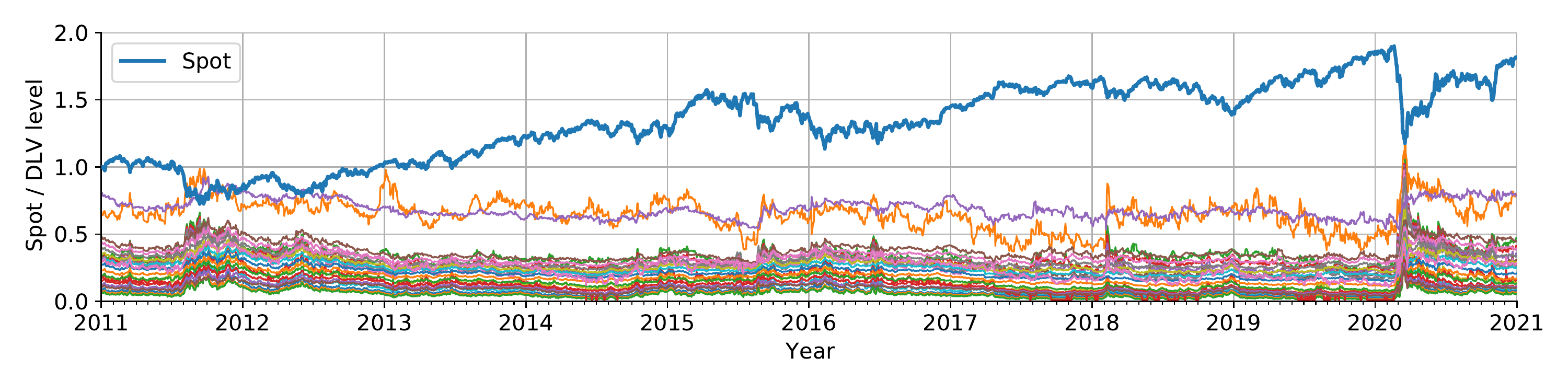}
    \caption{Normalized Eurostoxx 50 spot level and DLVs for maturities $\mathcal{T} = (60, 120) $ (quoted in business days) and relative strikes $\mathcal{K}=(0.7, 0.75, \dots, 1.0, \dots, 1.25, 1.3)$.}
    \label{fig:spot_dlvs}
\end{figure*}

We emphasise that the approach presented in this paper is a broad framework to construct a risk-neutral market simulator from real equity option market data. Our approach is not limited to the application of normalizing flows. Normalizing flows are solely used for the approximation of a conditional generative density.

\section{Background}
In the following subsections, we introduce the three main building blocks which are required to be able to construct a risk-neutral market simulator. We begin by introducing neural spline flows for approximating conditional generative densities. Afterwards, we review the construction of the equity option market simulator introduced in \cite{wiese2021multiasset} and derive the objective function. In the last subsection, we introduce entropy-based risk-neutral densities \cite{buehler2022deep}. 

\subsection{Neural spline flows for time series modeling}
Normalizing flows \cite{papamakarios2019normalizing} are \textit{diffeomorphisms}\footnote{Differentiable and bijective functions.} that are constructed by the means of neural networks. Being bijective allows the approximation of densities through the application of the density transformation theorem. In this paper, we will consider neural splines flows \cite{durkan2019neural} which are a specific subclass of normalizing flows that recently got popular due to the ability to universally approximate densities and have an analytically tractable inverse. 

Recently, neural splines with linear interpolation were proposed to approximate generative densities for time series data \cite{wiese2021multiasset}. For completeness, we illustrate the construction of neural splines with linear interpolation for time series below. Assume that $(X_{t})_{t \in \N} \sim p $ is a discrete-time Markov process taking, without loss of generality, values in $[0, 1]^d$. Furthermore, let $T_\eta=(T_{\eta, 1}, \dots, T_{\eta, d}): [0, 1]^d \times [0, 1]^d \to [0, 1]^d$ be a differentiable triangular increasing map in the second component (cf. \cite{bogachev2005triangular}) parametrised by $\eta \in H$. Then the \emph{generated process} can be defined for $j=1, \dots, d$ as 
\begin{equation*}
    X_{t+1, j} = T_{\eta, j}(U_j; X_{t}, X_{t+1, :(j-1)}) \sim p_\eta 
\end{equation*}
where $U = (U_1, \dots, U_d) \sim \mathcal{U}([0, 1]^d)$ is a uniformly distributed random variable on the unit hypercube. Since the function $T_\eta$ is assumed to be bijective, the density transformation theorem can be applied for $j=1, \dots, d$ to obtain the conditional density induced by $T$
\begin{align*}
    &p_\eta(x_{t+1, j}| x_t, x_{t+1, :j-1}) \\&= p(u_{t+1, j})| {T_{\eta, j}}'(u_{t+1, j}; x_t, x_{t+1, :j-1})|^{-1}
\end{align*}
where $u_{t+1, j} =T_{\eta, j}^{-1}(x_{t+1, :j-1}; x_{t}, x_{t+1, :j-1})$ is the latent variable conditional on the state $x_{t}$ and the next day partial state $x_{t+1, :j-1}$, and ${T_{\eta, j}}'$ denotes the derivative of $T_{\eta, j}$ with respect to $u$. 

To make the derivative analytically tractable while maintaining monotonicity in $T_{\eta, j}, j =1, \dots, d$ different interpolation schemes have been proposed \cite{durkan2019cubic, durkan2019neural}. In this paper, we use linear interpolation and define the conditional density for a fixed number of bins $B \in \N$ as  
\begin{equation*}
    p_\eta(x_{t+1, j} | x_t, x_{t+1, :(j-1)}) = \prod_{k=1}^B p_k^{c_k} 
\end{equation*}
where $(p_k)_{k}$ is non-negative, normalized, and $c_k$ is defined for $k = 1, \dots, B$ as
\begin{equation*}
    c_k = \left\lbrace
    \begin{array}{ll}
        1, \quad \textrm{if} \ x_{t+1} \in [\frac{k-1}{B}, \frac{k}{B}]  \\
        0, \quad \textrm{else}
    \end{array}
    \right. 
\end{equation*}
Neural splines approximate the conditional density $(p_k)_k$ with with a neural network taking values in $\R^B$. The output is then transformed via the softmax function and normalized to obtain a density. 
\begin{figure*}
    \centering
    \includegraphics[width=\textwidth]{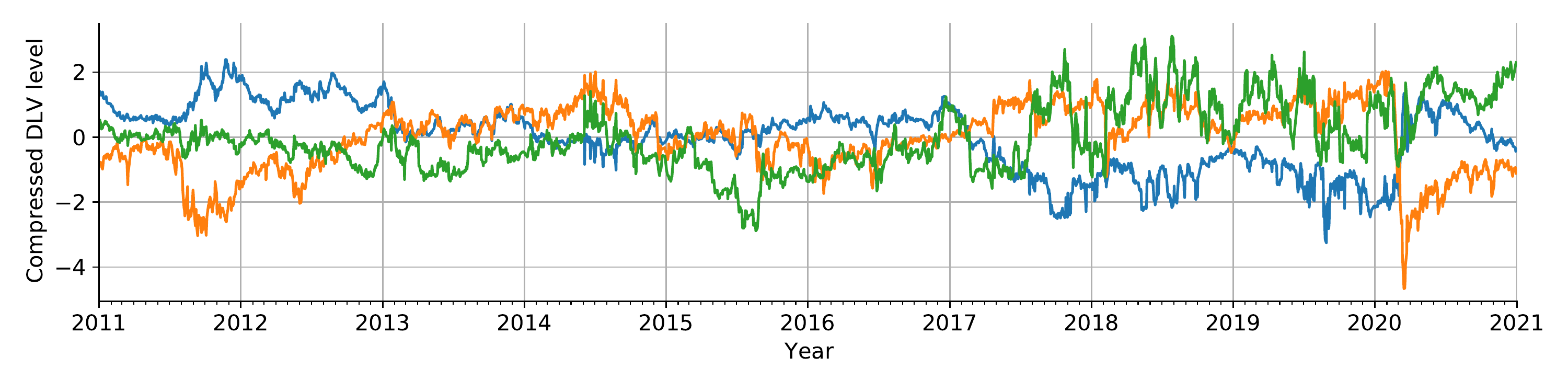}
    \caption{Encoded DLVs obtained through a calibrated autoencoder; $\mathcal{T} = (60, 120) $ (quoted in business days), $\mathcal{K}=(0.7, 0.75, \dots, 1.0, \dots, 1.25, 1.3)$.}
    \label{fig:compressed_dlvs}
\end{figure*}

Due to the tractability of the conditional density the flow $T_{\eta}$ can be approximated using \emph{conditional} Kullback-Leibler (KL) divergence, or equivalently, the negative log-likelihood (NLL)
\begin{align}
    &\KL(p\|p_\eta)
    \notag 
    = \E_{p}\left(\E_{p}\left(\ln \dfrac{p(X_{t+1} | X_t=\hat{X}_t)}{p_\eta(X_{t+1} | X_t=\hat{X}_t)} \right) \big | \hat{X}_t \right) 
    \\
    \label{eq:nll_flows}
    &= -\E_{p}\left(\E_{p}\left(\ln p_\eta(X_{t+1} | X_t=\hat{X}_t) \big | \hat{X}_t\right)\right) + const
\end{align}
Note that here we use a conditional version of the KL-divergence since our objective is to approximate the conditional density of the market process. In practice, only a single time series $ (x_t)_{t=1}^{T} \sim p$ may be available. In this case, \eqref{eq:nll_flows} will be approximated via Monte Carlo: 
\begin{equation*}
    \hat{J}(\eta) = (T-1)^{-1} \sum_{t=1}^{T-1} -\ln p_\eta(x_{t+1} | x_t)
\end{equation*}

\subsection{Equity option market simulation}
In this subsection, we briefly outline the construction of the market simulator. For an in-depth introduction we direct the reader to \cite{wiese2021multiasset}. 

Let $(\Omega, \mathcal{F} = (\mathcal{F}_t)_{t \in \N}, \mathbb{P})$ be a filtered probability space. We refer to the probability measure $\P$ as the \emph{physical probability measure} and by $p$ we refer to the assosciated \emph{physical density}. Denote by $X=(S, C): \Omega \times \N \to \R_{>0}\times\R^{mn}_{>0}$ an adapated discrete-time Markov process, i.e. we assume for any $s, t \in \N, t > s $ that 
\begin{equation*}
    p(x_{t} | \mathcal{F}_s) = p(x_{t} | x_s) %\ .
\end{equation*}
The first component $S=(S_t)_{t \in \N}$ represents the spot price of the underlying which we assume to take the form $S_{t+1} = S_t (1 + R_{t+1})$ where $R_{t+1}$ is independent from $S_t$. The second component $C=(C_t)_{t \in \N}$ represents an $mn$-dimensional grid of call prices defined on a floating grid of \emph{time to maturities} $\mathcal{T} = (\tau_1, \dots, \tau_m)$ and \emph{relative strikes} $\mathcal{K} = (k_1, \dots, k_n)$ defined around the unit forward. Thus, the call price grid at any time $t \in \N$ is defined for $i = 1, \dots, m$ and $j = 1, \dots n$ so that $C_{t, (i-1)*n + j}$ is the price of the option with payoff $( S_{t +\tau_i} / S_t - k_j)^+$. In what follows we refer to $X$ as the \emph{market process}. 

\subsubsection{Discrete local volatilities}
In order to guarantee the absence of \emph{static arbitrage} (riskless profits) realized grid prices $c_t \in \R^{mn}_{> 0}$ need to satisfy ordering constraints such as non-negativity, monotonicity in time and convexity in strike (cf. \cite{gatheral2014arbitrage}). We therefore represent the grid prices by \emph{discrete local volatilities} (\emph{DLVs}) \cite{DLV} as an arbitrage-free parametrisation of the considered price grid.\footnote{No static arbitrage is a hard requirement: if there exists static arbitrage under $\P$, then no equivalent risk-neutral measure exists.} DLVs take a local volatility-inspired parametrisation and can be seen as a discrete version of Dupire's famous local volatility \cite{dupire1994pricing}. Most importantly, the mapping from non-negative DLVs to the call price grid is a bijective map $\Phi: \R^{mn}_{>0} \to \R^{mn}$ which will be necessary to construct the manifold flow \eqref{eq:manifold_flow}. We denote the unique corresponding grid of DLVs by $\Sigma_t = \Phi^{-1}(C_t)$. 

\subsubsection{Interpolating the call price grid}
The market process $X$ was constructed using floating grid prices to obtain a stationary representation of the call prices. In order to obtain real-world prices with fixed expiries and strikes we interpolate the floating price grid using a bilinear interpolation\footnote{Note that any interpolation scheme that maintains monotonicity in time and convexity in strike will not introduce static arbitrage and therefore is permitted.} in the prices. For any $t \in \N$ denote the bilinear interpolated grid at maturity $\tau \in [\min(\mathcal{T}), \max(\mathcal{T})]$ and relative strike $k \in [\min(\mathcal{K}), \max(\mathcal{K})]$ as $\tilde{C}_{t}(\tau, k)$. Then the price of a call bought on day $t$ at maturity-strike $(\tau, k) \in \mathcal{T}\times\mathcal{K}$ has a price on day $t+1 \in \N$ of
\begin{equation*}
    C_{t+1, \tau, k}(S_{t+1}, S_t) = S_{t+1} / S_t \tilde{C}_{t+1}(\tau-1, k S_t / S_{t+1})
\end{equation*}
Note that a spot move $ s_{t+1} / s_{t} $ greater than $k / \min(\mathcal{K})$ or smaller than $ k / \max(\mathcal{K})$ will not make the interpolation feasible due to the limited relative strike range. We treat this problem by introducing boundary relative strikes and maturities and interpolate to the option's intrinsic value \cite{DLV}. Using the option's interpolated value we can compute the difference in price of the call bought at time $t$ at with maturity-strike pair $(\tau, k)$ as 
\begin{equation*}
    dC_{t+1} = C_{t+1, \tau, k}(S_{t+1}, S_t) - \tilde{C}_{t}(\tau, k)%.
\end{equation*}

\subsubsection{Manifold assumption}
We assume grid prices $c_t \in \R^{mn}_{>0}$ to have support on a low-dimensional manifold. More precisely, we assume that there exists a \emph{latent dimension} $l \ll mn$, an injective map $\psi: \R^{l} \to \R^{mn}_{>0}$ such that for any price grid $c_t\in \R^{mn}_{>0}$ with positive density there exists a unique representation $\sigma_t $ such that $ c_t = \psi(\sigma_t)$ holds $\P$-almost surely. The manifold assumption is justified by the high observed correlation in DLV levels as well as their returns (cf. \cite{wiese2021multiasset}). 

Under the call price manifold assumption we can construct for any states $x_{t}, x_{t+1} \in \Image(\psi)$ in the image set of $\psi$ the manifold flow \cite{brehmer2020flows, gemici2016normalizing}
\begin{equation}
    \label{eq:manifold_flow}
    p(x_{t+1} | x_t) = p(y_{t+1}| y_t) f(y_{t+1}) 
\end{equation}
where $y_{v} = (s_{v}, \sigma_v=\psi^{-1}(c_{v})), v \in \lbrace t, t+1\rbrace$ is the \emph{compressed state representation} and 
\begin{equation*}
    f(y_{t+1}) = \left(\det J_\psi(\sigma_{t+1})^T J_\psi(\sigma_{t+1})\right)^{-\frac{1}{2}}
\end{equation*}
where $J_\psi$ is the Jacobian of $\psi$. 
By leveraging \eqref{eq:manifold_flow} one can show that minimzing the NLL on the level of the spot and call prices \eqref{eq:nll_flows} is equivalent to minimizing the NLL on the compressed state space 
\begin{align*}
    J(\eta) &= -\E_{p}\left(\E_{p}\left(\ln p_\eta(X_{t+1} | X_t =\hat{X}_{t}) \right) | \hat{X}_{t}\right) 
    \\
    &= -\E_{p}\left(\E_{p}\left(\ln p_\eta(Y_{t+1} | Y_t = \hat{Y}_t) \right) | \hat{Y}_{t}\right) + const 
\end{align*}
This motivates our choice to approximate market simulators in the compressed space. 

\subsection{Entropy-based risk-neutral densities}
A central concept in mathematical finance is that of a \emph{risk-neutral probability measure}, $\Q$ which is equivalent to $\P$, under which the price of all tradeable instruments is the discounted expectation of the instrument's payoff. Throughout this article, we will assume for simplicity that the risk-free rate is zero which implies that $X$ is a martingale under $\Q$. It is well known that if the physical measure is free from arbitrage then such a risk-neutral measure must exist, although it will not in general be unique.

To construct such a measure, recall that $X$ represents instruments which are liquidly traded in a market. Consider a \textit{trading action} $a_t \equiv a(X_t) \in \R^{1+mn}$, and a \textit{gains} of taking this action over a single timestep of $G(X_{t+1}, X_t) = a_t \cdot (X_{t+1} - X_t)$. Consider the strictly concave, strictly increasing utility function $u(x) = ( 1 - e^{-\lambda x})/\lambda$ -- known as the exponential utility, or entropy. Let $a^*$ be a solution of the optimization problem 

\begin{equation}
    \label{eq:utility_optimization}
    \sup_a \E \left[ u(G(X_{t+1}, x_t)) \right] %,
\end{equation}
and let $G^*$ be the associated optimal gains. Then we can define a measure $\Q$ via the conditional density, 
\begin{align}
    \label{eq:memm}
    \qdens (X_{t+1}| X_t) &= w(X_{t+1}|X_t) \ \pdens (X_{t+1}|X_t) \\ %\ , 
    w(X_{t+1}|X_t) &= \frac{e^{-\lambda G^*(X_{t+1}, X_t)}}{\E \left[e^{-\lambda G^*(X_{t+1}, X_t)} | X_t \right]} %\ .
\end{align}

Using this density, we have the following result.

\begin{proposition}
The conditional density $\qdens$ is an equivalent martingale density, i.e.  $\E_\qdens [X_{t+1}| X_t ] = X_t$. Furthermore, $\qdens$ minimizes the KL divergence to $\pdens$ over all equivalent martingale densities, and $\Q$ is hence called the \textit{minimal entropy martingale measure} \cite{FRITTELLI}.
\end{proposition}

The measure $\Q$ is one example of a broader class of utility-based risk-neutral measures, see \cite{buehler2022deep} for further details. Due to the positivity of the exponential function, $w > 0$, ensuring that $\P$ and $\Q$ are equivalent, and the concavity of $u$ means that $w$ is a decreasing function of $G^*$. Furthermore, the change of measure is normalized, i.e. $\E(w(X_{t+1} | X_t=x_t)) = 1$, ensuring that $q$. The financial interpretation is that by reweighting future states of the world proportional to their marginal utility under the physical measure, outcomes which were profitable under $\P$ become less likely under $\Q$, resulting in a measure in which all profit opportunities have been removed and all instruments must have zero drift.  

\begin{figure*}[htp]
    \centering 
    \includegraphics[width=\textwidth]{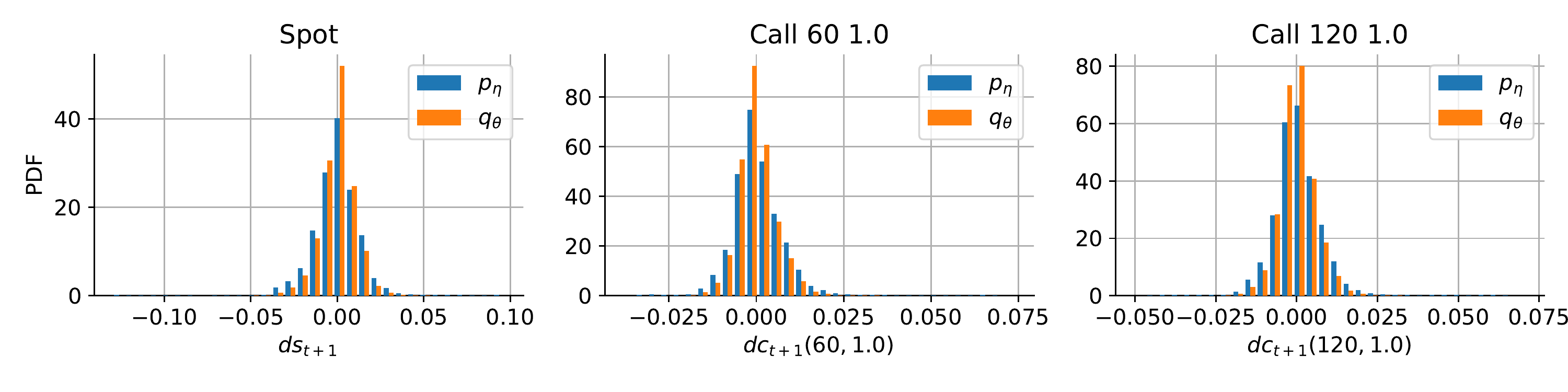}
    \caption{Empirical histograms of the spot price, 60D and 120D at-the-money call price sampled under the approximated real-world density $p_\eta$ (blue) and risk-neutral density $q_\theta$ (orange).}
    \label{fig:spot_call_density}
\end{figure*}

\section{Problem formulation}
Our aim in this paper is to calibrate a \emph{model density} $\qdens_\theta$ that is close to the ground-truth risk-neutral density $\qdens$ with respect to the KL divergence, by using only samples from the physical density $\pdens$. Note that our framework is very general, and applies to any density estimator, and in this paper we focus specifically on neural spline flows. Ignoring constants for clarity, our objective is therefore the following:  

\begin{equation}
    J(\theta) = - \ \mathbb{E}_\qdens \left(\mathbb{E}_\qdens \left( \ln \qdens_\theta(\X_{t+1} | X_t = \hat{X}_t) \right) | \hat{X}_{t}\right) %\ .
\end{equation}

Note that we minimize this objective by minimizing the inner conditional NLL for all states $\hat{X}_t$, so the outer expectation can be taken over the equivalent physical measure. Furthermore, since we can neither evaluate nor sample from the target density $\qdens$, we apply the change of measure transform above, so that this objective is equivalent to 

\begin{equation}
    \label{eq:calibration_problem_*}
    \tag{$\ast$}
    J(\theta) = - \E_p \left( \E_p \left( w(X_{t+1}|X_t=\hat{X}_t)\ln  \qdens_\theta(X_{t+1} |  X_t = \hat{X}_t) \right) | \hat{X}_{t}\right) %\ .
\end{equation}

Due to the assumption that the decoder $\psi: \R^l \to \R^{mn}_{>0}$ is injective the minimization problem \eqref{eq:calibration_problem_*} is equivalent to solving following the calibration problem on the reduced space \eqref{eq:manifold_flow}. Further note that for any $x_{t}, x_{t+1} \in \Image(\psi)$ the change of measure $w(x_{t+1} | x_t) = \qdens(x_{t+1} | x_t) / \pdens(x_{t+1} | x_t)$ on the observed space coincides with the change of measure on the reduced space:
\begin{align*}
    w(\x_{t+1} | \x_t) 
    %= \dfrac{\qdens(x_{t+1} | x_t)}{\pdens(x_{t+1} | x_t)} 
    = \dfrac{\qdens(y_{t+1} | y_t) f(\sigma_{t+1})}{\pdens(y_{t+1} | y_t) f(\sigma_{t+1})} 
    = w(\y_{t+1} | \y_t) %\ . 
\end{align*}
Hence, we may write our objective as a weighted NLL on the reduced space:
\begin{equation}
    \label{eq:calibration_problem_**}
    \tag{$\ast\ast$}
    J(\theta) = - \E_p \left(\E_p\left(  w(Y_{t+1}|Y_t= \hat{Y}_t) \ln \qdens_\theta(Y_{t+1} | Y_t = \hat{Y}_t) \right) | \hat{Y}_{t}\right) %\ .
\end{equation}

In the real world, both calibration problems \eqref{eq:calibration_problem_*} and \eqref{eq:calibration_problem_**} are not feasible. We only observe a single realization $(\tilde{\x}_t)_{t \in I}$ sampled from the physical density $\pdens$ over some time horizon $I \coloneqq \lbrace 0, \dots, T \rbrace$. Furthermore, we do not observe the call price function $\psi$. We instead only have access to the \emph{empirical physical density}  
    $\tilde{\pdens}(\x_{t+1} | \x_{t}) = \sum_{t\in \tilde{I}} \delta_{(\tilde{\x}_{t+1}, \tilde{\x}_{t})}((\x_{t+1} , \x_t)) $
where $\delta$ denotes the Dirac delta function, and the corresponding \emph{empirical reduced physical density}. This in particular makes learning the change of measure weights difficult when done directly from the empirical densities, since the optmization would be done on only single samples from each condition. To overcome these challenges, we use the following data augmentation method. 

\subsection{Augmentation with the physical market simulator}
To get good estimates of the weights $w(x_{t+1}|x_t)$ we augment the empirical distribution by first training a flow model density $p_\eta$ that minimizes the KL divergence to the physical density. We then train the model risk-neutral density to target the weighted model density. This is justified by the following proposition.   

\begin{proposition}
Let $\pdens_\eta$ be a model density for the physical probability measure, such that $\KL(\pdens \| \pdens_\eta) = 0$ and let $\qdens_\eta = w \pdens_\eta$. Let $T_\theta$ be a Markovian market simulator such that the induced density $\qdens_\theta$ satisifies $\KL(\qdens_\eta \| \qdens_\theta) = 0$. Then $\qdens_\theta = \qdens$ and the market simulator is risk-neutral.  
\end{proposition}

\begin{proof}
Since zero KL divergence implies (almost sure) equality of the distributions, we have $\qdens_\theta = \qdens_\eta = w \pdens_\eta = w \pdens = \qdens$. Hence for all $t$, we must have $\E_{\qdens_\theta} [X_{t+1}|X_t] = X_t$, which from the Markov assumption, implies that the simulator is risk-neutral. 
\end{proof}

Therefore, we may train the risk-neutral market simulator on samples from the physical simulator via the NLL-based objective:
\begin{align*}
    J(\theta )= - \E_{\pdens_\eta} \left( w(\X_{t+1}| \x_t) \ln \qdens_\theta(\X_{t+1} | {\x}_t)\right) %\ .
\end{align*}

\paragraph{Change of measure approximation}
We further use samples from the physical simulator to approximate the change of measure weights by solving \eqref{eq:utility_optimization}. This is equivalent to the convex optimization
\begin{equation*}
    L(a) = \mathbb{E} (\exp(-a \cdot dX_{t+1}) | X_t=x_t)
\end{equation*}
where $ dX_{t+1} = X_{t+1} - X_{t}$ denotes the change in the instrument's value. The Hessian is given as 
\begin{equation*}
    \nabla_{a}^2 L(a) = \lambda^2\mathbb{E}(\exp(-\lambda a \cdot dX_{t+1}) dX_{t+1} \otimes dX_{t+1}| X_t=x_t) %\ .
\end{equation*}
Practically, the utility loss function is approximated via MC: 
\begin{equation}
    \label{eq:optimal_actions}
    \hat{L}(a) = N^{-1}\sum_{j=1}^N\exp(-\lambda a \cdot dx_{t+1}^{(j)}) %\ ,
\end{equation}

\begin{equation*}
    \nabla_{a}^2 \hat{L}(a) = \lambda^2 N^{-1}\sum_{j=1}^N \exp(-\lambda a \cdot dx_{t+1}^{(j)}) dx_{t+1}^{(j)} \otimes dx_{t+1}^{(j)} %\ ,
\end{equation*}
where $\lbrace x_{t+1}^{(j)}\rbrace_{j=1}^N \sim \pdens_\eta(\cdot | x_t)$ is a sample generated using the physical market simulator. 

The weights $w(x_{t+1}|x_t)$ are them computed directly using the converged $a^*$. An example of the weights obtained across different samples is shown in \autoref{fig:measure_change}. The distribution of the weights illustrated here indicates another clear advantage of the approach: sampling efficiency. Since many of the weights will be less than 1, by applying the change of measure on the conditional density rather than the joint of the entire path, we avoid degenerate weights, and therefore avoid using sample paths from $\pdens$ with very low density under $\qdens$.  

Putting these steps together, we summarize the approach in \autoref{alg:algo}. 

\begin{algorithm}[htp]
    \caption{Risk-neutral simulator training}\label{alg:algo}
    \begin{algorithmic}
    \REQUIRE {Samples $x_t$, $t=1, \ldots, T$ from historical data, trained physical simulator $\pdens_\eta$}
    \FOR {$t=1, \ldots, T$}
        \STATE {Generate $N$ samples from $\pdens_\eta(\cdot |x_t)$}
        \STATE {Solve \eqref{eq:optimal_actions} via convex optimization}
        \STATE {Calculate $w(x_{t+1}^{(j)}|x_t)$ for $j=1,\ldots,N$}
    \ENDFOR
    \STATE {Train risk-neutral simulator on $NT$ samples via the loss
            $$ \hat{J}(\theta) = - \frac{1}{NT} \sum_{t=1}^T \sum_{j=1}^N w(x_{t+1}^{(j)} |x_t) \ln \qdens_\theta (x_{t+1}^{(j)} |x_t) $$}
    \end{algorithmic}
\end{algorithm}

\section{Numerical results}
In this section, we demonstrate the efficacy of our flow-based risk-neutral market simulator to generate samples which are martingales but nonetheless demonstrate similar properties to the physical distribution. We consider the Eurostoxx 50 from January 2011 to December 2020 for a total of 2543 business days, and simulate spot and DLVs / options defined at relative strikes $\mathcal{K}=(0.7, 0.75, \dots, 1.0, \dots, 1.25, 1.3)$ and maturities $\mathcal{T} = (60, 120)$ days. The corresponding set of DLVs as well as the normalized Eurostoxx 50 spot price are displayed in \autoref{fig:spot_dlvs}. For clarity, we then simulate prices of three tradeable instruments: spot and the at-the-money call options with 60D and 120D expiry. 

\paragraph{Approximating the simulator under $\mathbb{P}$}
The real-world market simulator was constructed by following \cite{wiese2021multiasset}. We first approximated the autoencoder to reduce the dimensionality of the 26-dimensional grid of DLVs to a three-dimensional representation displayed in \autoref{fig:compressed_dlvs}. Afterwards, the flow was calibrated on the compressed representation using four neural networks (one for each dimension of the compressed state and one for the spot process). Each network had three layers, 64 hidden dimensions and a 64-dimensional output layer. To avoid overfitting we retained 20\% of the samples in the test set and applied early stopping \cite{caruana2001overfitting, prechelt1998early}. 

\begin{figure}[htp]
    %\centering 
    \includegraphics[width=0.5\textwidth]{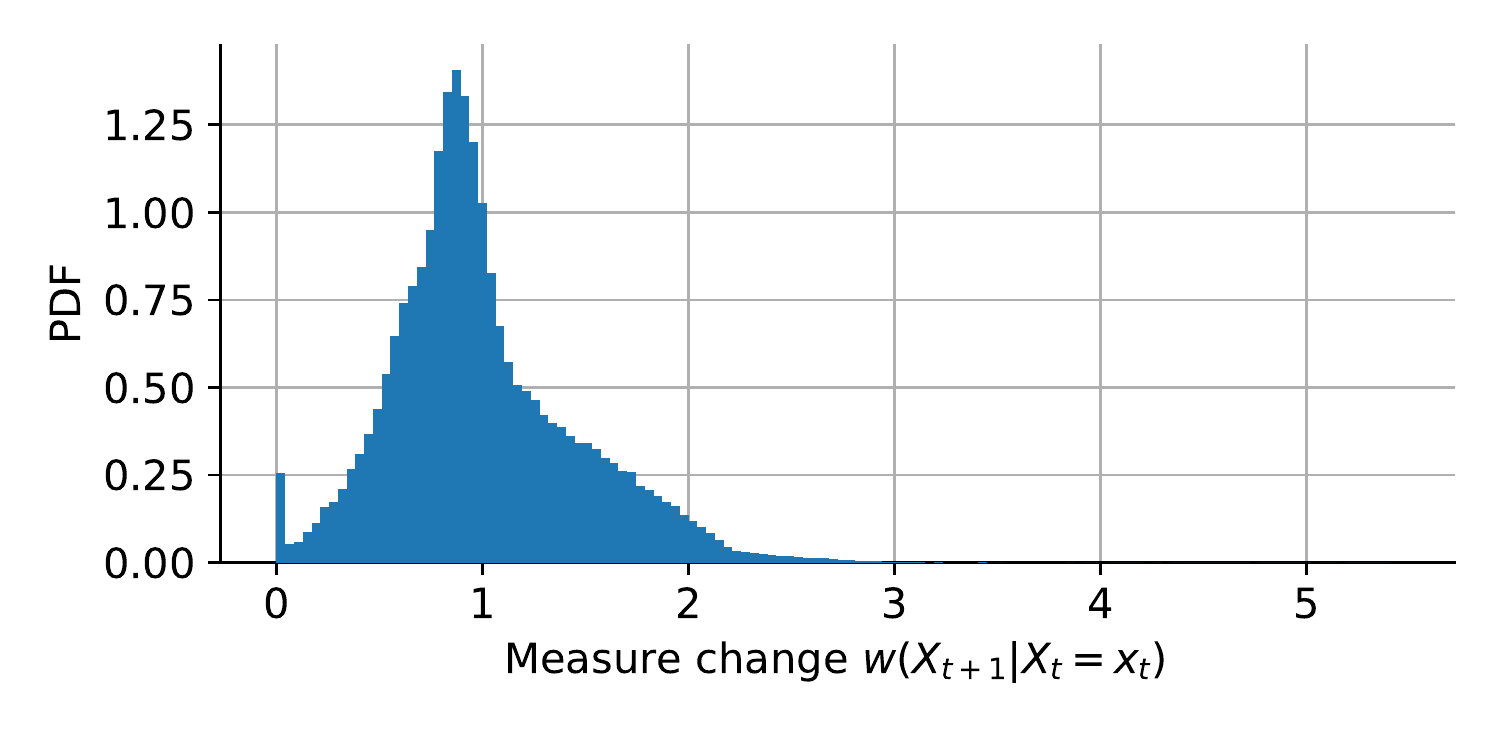}
    \caption{Distribution of the change of measure $w(X_{t+1} | X_t = x_t)$ calibrated on the considered three instruments.}
    \label{fig:measure_change}
\end{figure}

\paragraph{Evaluating the drift}
We test the risk-neutrality of the simulator in two ways. First, we generate a large number $N= 262.144 (=2^{18})$ of samples from both $\pdens_\eta(\cdot|x_t)$ and $\qdens_\theta(\cdot|x_t)$ and compare the MC drift $1/N \sum_j x^{(j)}_{t+1} - x_t$ for all three instruments, between the two simulators. \autoref{table:drift} shows clearly the reduction in drift in the marginal distributions, in both absolute and relative terms.

\begin{table}
\centering
\resizebox{\columnwidth}{!}{
\begin{tabular}{lrrr}
\toprule
{} &      Spot &  Call (60, 1.0) &  Call (120, 1.0) \\
\midrule
$\P$ drift      &  0.000041 &     0.000923 &      0.000773 \\
$\Q$ drift     & -0.000019 &    -0.000001 &     -0.000009 \\
$\P$ drift (\%) &  0.0041 &     2.21 &      1.21 \\
$\Q$ drift (\%) &  0.0019 &     0.0027 &      0.0156 \\
Ratio                  &  2.16 &   827.08 &     82.75 \\
\bottomrule
\end{tabular}
}
\vspace{8pt}
\caption{Numerical drifts obtained under the real-world and near risk-neutral approximated conditional density.}
\label{table:drift}
\end{table}

Secondly, we check the joint, we check that there are no profitable trading opportunities under $\qdens_\theta$. Since $u(x) \leq x$ for all $x$, it follows that $\E(u(a\cdot dX)) \leq \E (a \cdot dX)$ and hence under a risk-neutral measure, the optimal action for the exponential utility is $a^*=0$. Figure \ref{fig:pnl_vs_spot} shows the distribution of PnL obtained by optimal actions against different spot shifts. Specifically we see how the optimal actions under $\pdens_\eta$ are able to profit from larger spot moves but make losses on small moves. This can be understood by comparing the Black-Scholes implied volatilities\footnote{Of course the market does not follow Black-Scholes dynamics, however it provides a useful and simple framework to understand the optimal trading action.} of the options with the realized volatility of the spot under $\pdens_\eta$. In this case, the options trade at implied volatilities of $21.3\%$ and $23.1\%$ whereas the realized volatility for the spot is $25.4\%$, making options cheap, and giving rise to a positive drift in the options - the optimal trading action being to buy options and hedge short with the spot. Under the risk-neutral density, this drift is corrected, as can be seen in the subtle reduction in the volatility in the spot returns demonstrated in \autoref{fig:spot_call_density}. The end result is that the optimal actions under $\qdens_\theta$ are effectively zero, with no PnL or utility obtainable.

\begin{figure}[htp]
    \centering 
    \includegraphics[width=0.5\textwidth]{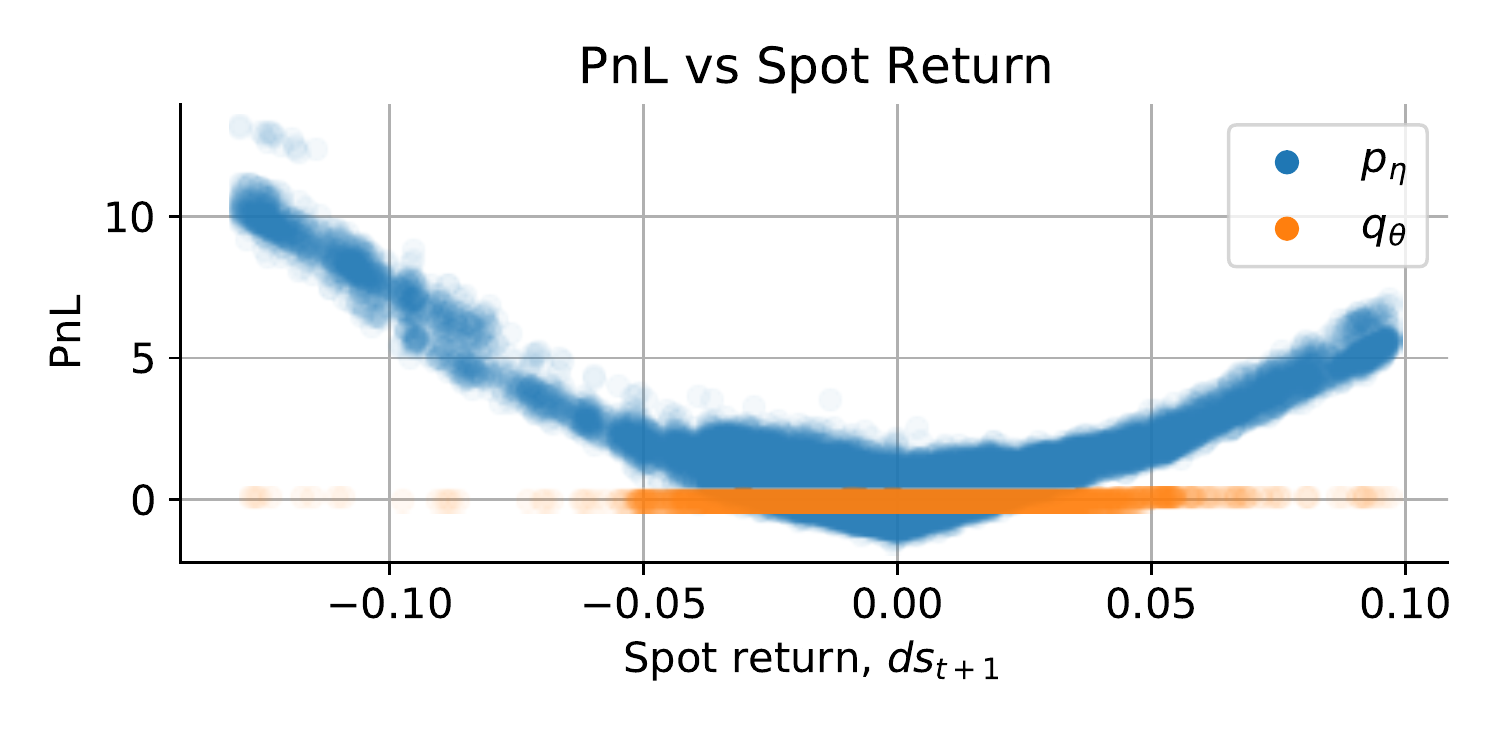}
    \caption{Scatter plot comparing the spot move $ds_{t+1}$ against the PnL $a^* \cdot dX_{t+1}$ under the approximated real-world density $p_\eta$ and the approximated \emph{near} risk-neutral density $q_\theta$.}
    \label{fig:pnl_vs_spot}
\end{figure}

\paragraph{Evaluating the conditional risk-neutral density in the compressed space}
We evaluated the conditional approximated real-world and risk-neutral density $ p_\eta(\cdot | y_{t})$ and  $ q_\theta(\cdot | y_{t})$ by applying a kernel estimator on the generated samples in the compressed space (see \autoref{fig:contour_plot}). Visually, we can observe that the change of measure alters the tails of the distribution. In particular, the tails of the spot distribution become narrower which resulted in a clear reduction in utility and PnL of the optimal volatility spread strategy $a^*$ (see \autoref{fig:pnl_vs_spot}). %Similar observations can be made when evaluating the conditional density of the tradeable instruments in \autoref{fig:spot_call_density}. 

\section{Conclusion}
In this paper, we have outlined a general framework for data-driven simulation of financial market data from a risk-neutral measure, by combining a calibrated market simulator for the physical measure $\P$, and a change of measure found by solving a convex optimization on the physical simualator. We then train a simulator for the risk-neutral measure $\Q$ using a weighted negative log likelihood objective.

We demonstrate the effectiveness of the method by using neural spline flows to estimate construct a simulator for markets of equity underlying and options and show that the resulting simulator is indeed risk-neutral, while still retaining relevant properties of the historical data.

We leave it as future work to construct the full risk-neutral market simulator from which one can sample and expand the set of tradable instruments to wider floating grids of relative strikes and maturities. Another direction of research is the application of manifold flows to guarantee the injectivity of the decoder by construction \cite{brehmer2020flows}.

\clearpage
\bibliography{references}

\begin{thebibliography}{26}
\providecommand{\natexlab}[1]{#1}

\bibitem[{Arribas, Salvi, and Szpruch(2020)}]{arribas2020sigsdes}
Arribas, I.~P.; Salvi, C.; and Szpruch, L. 2020.
\newblock Sig-SDEs model for quantitative finance.
\newblock arXiv:2006.00218.

\bibitem[{Bogachev, Kolesnikov, and Medvedev(2005)}]{bogachev2005triangular}
Bogachev, V.~I.; Kolesnikov, A.~V.; and Medvedev, K.~V. 2005.
\newblock Triangular transformations of measures.
\newblock \emph{Sbornik: Mathematics}, 196(3): 309.

\bibitem[{Brehmer and Cranmer(2020)}]{brehmer2020flows}
Brehmer, J.; and Cranmer, K. 2020.
\newblock Flows for simultaneous manifold learning and density estimation.
\newblock arXiv:2003.13913.

\bibitem[{Buehler et~al.(2019)Buehler, Gonon, Teichmann, and Wood}]{DH}
Buehler, H.; Gonon, L.; Teichmann, J.; and Wood, B. 2019.
\newblock Deep Hedging.
\newblock \emph{Quantitative Finance}, 0(0): 1--21.

\bibitem[{Buehler et~al.(2020)Buehler, Horvath, Lyons, Arribas, and
  Wood}]{horvath2020}
Buehler, H.; Horvath, B.; Lyons, T.; Arribas, I.~P.; and Wood, B. 2020.
\newblock A Data-driven Market Simulator for Small Data Environments.
\newblock arXiv:2006.14498.

\bibitem[{Buehler et~al.(2022)Buehler, Murray, Pakkanen, and
  Wood}]{buehler2022deep}
Buehler, H.; Murray, P.; Pakkanen, M.~S.; and Wood, B. 2022.
\newblock Deep Hedging: Learning to Remove the Drift under Trading Frictions
  with Minimal Equivalent Near-Martingale Measures.
\newblock arXiv:2111.07844.

\bibitem[{Buehler and Ryskin(2016)}]{DLV}
Buehler, H.; and Ryskin, E. 2016.
\newblock Discrete Local Volatility for Large Time Steps (Short Version).

\bibitem[{Caruana, Lawrence, and Giles(2001)}]{caruana2001overfitting}
Caruana, R.; Lawrence, S.; and Giles, L. 2001.
\newblock Overfitting in neural nets: Backpropagation, conjugate gradient, and
  early stopping.
\newblock \emph{Advances in neural information processing systems}, 402--408.

\bibitem[{Cohen, Reisinger, and Wang(2021)}]{cohen2021arbitrage}
Cohen, S.~N.; Reisinger, C.; and Wang, S. 2021.
\newblock Arbitrage-free neural-SDE market models.
\newblock arXiv:2105.11053.

\bibitem[{Cohen, Reisinger, and Wang(2022)}]{cohen2022estimating}
Cohen, S.~N.; Reisinger, C.; and Wang, S. 2022.
\newblock Estimating risks of option books using neural-SDE market models.
\newblock \emph{arXiv preprint arXiv:2202.07148}.

\bibitem[{Cuchiero, Khosrawi, and Teichmann(2020)}]{cuchiero2020generative}
Cuchiero, C.; Khosrawi, W.; and Teichmann, J. 2020.
\newblock A generative adversarial network approach to calibration of local
  stochastic volatility models.
\newblock \emph{Risks}, 8(4): 101.

\bibitem[{De~Meer~Pardo, Schwendner, and Wunsch(2021)}]{de2021tackling}
De~Meer~Pardo, F.; Schwendner, P.; and Wunsch, M. 2021.
\newblock Tackling the exponential scaling of signature-based GANs for
  high-dimensional financial time series generation.
\newblock \emph{Available at SSRN 3942764}.

\bibitem[{Dupire(1994)}]{dupire1994pricing}
Dupire, B. 1994.
\newblock Pricing with a smile.
\newblock \emph{Risk}, 7(1): 18--20.

\bibitem[{Durkan et~al.(2019{\natexlab{a}})Durkan, Bekasov, Murray, and
  Papamakarios}]{durkan2019cubic}
Durkan, C.; Bekasov, A.; Murray, I.; and Papamakarios, G. 2019{\natexlab{a}}.
\newblock Cubic-spline flows.
\newblock \emph{arXiv preprint arXiv:1906.02145}.

\bibitem[{Durkan et~al.(2019{\natexlab{b}})Durkan, Bekasov, Murray, and
  Papamakarios}]{durkan2019neural}
Durkan, C.; Bekasov, A.; Murray, I.; and Papamakarios, G. 2019{\natexlab{b}}.
\newblock Neural spline flows.
\newblock \emph{Advances in Neural Information Processing Systems}, 32:
  7511--7522.

\bibitem[{Frittelli(2000)}]{FRITTELLI}
Frittelli, M. 2000.
\newblock The minimal entropy martingale measure and the valuation problem in
  incomplete markets.
\newblock \emph{Mathematical finance}, 10(1): 39--52.

\bibitem[{Gatheral and Jacquier(2014)}]{gatheral2014arbitrage}
Gatheral, J.; and Jacquier, A. 2014.
\newblock Arbitrage-free SVI volatility surfaces.
\newblock \emph{Quantitative Finance}, 14(1): 59--71.

\bibitem[{Gemici, Rezende, and Mohamed(2016)}]{gemici2016normalizing}
Gemici, M.~C.; Rezende, D.; and Mohamed, S. 2016.
\newblock Normalizing flows on riemannian manifolds.
\newblock \emph{arXiv preprint arXiv:1611.02304}.

\bibitem[{Gierjatowicz et~al.(2020)Gierjatowicz, Sabate-Vidales, Siska,
  Szpruch, and Zuric}]{gierjatowicz2020robust}
Gierjatowicz, P.; Sabate-Vidales, M.; Siska, D.; Szpruch, L.; and Zuric, Z.
  2020.
\newblock Robust pricing and hedging via neural SDEs.
\newblock \emph{Available at SSRN 3646241}.

\bibitem[{Ni et~al.(2021)Ni, Szpruch, Sabate-Vidales, Xiao, Wiese, and
  Liao}]{ni2021sig}
Ni, H.; Szpruch, L.; Sabate-Vidales, M.; Xiao, B.; Wiese, M.; and Liao, S.
  2021.
\newblock Sig-Wasserstein GANs for Time Series Generation.
\newblock \emph{arXiv preprint arXiv:2111.01207}.

\bibitem[{Ni et~al.(2020)Ni, Szpruch, Wiese, Liao, and Xiao}]{hao2020}
Ni, H.; Szpruch, L.; Wiese, M.; Liao, S.; and Xiao, B. 2020.
\newblock Conditional Sig-Wasserstein GANs for Time Series Generation.
\newblock arXiv:2006.05421.

\bibitem[{Papamakarios et~al.(2019)Papamakarios, Nalisnick, Rezende, Mohamed,
  and Lakshminarayanan}]{papamakarios2019normalizing}
Papamakarios, G.; Nalisnick, E.; Rezende, D.~J.; Mohamed, S.; and
  Lakshminarayanan, B. 2019.
\newblock Normalizing flows for probabilistic modeling and inference.
\newblock \emph{arXiv preprint arXiv:1912.02762}.

\bibitem[{Prechelt(1998)}]{prechelt1998early}
Prechelt, L. 1998.
\newblock Early stopping-but when?
\newblock In \emph{Neural Networks: Tricks of the trade}, 55--69. Springer.

\bibitem[{Wiese et~al.(2019)Wiese, Bai, Wood, and Buehler}]{wiese2019deep}
Wiese, M.; Bai, L.; Wood, B.; and Buehler, H. 2019.
\newblock Deep hedging: learning to simulate equity option markets.
\newblock \emph{arXiv preprint arXiv:1911.01700}.

\bibitem[{Wiese et~al.(2020)Wiese, Knobloch, Korn, and
  Kretschmer}]{wiese2020quant}
Wiese, M.; Knobloch, R.; Korn, R.; and Kretschmer, P. 2020.
\newblock Quant gans: Deep generation of financial time series.
\newblock \emph{Quantitative Finance}, 20(9): 1419--1440.

\bibitem[{Wiese et~al.(2021)Wiese, Wood, Pachoud, Korn, Buehler, Murray, and
  Bai}]{wiese2021multiasset}
Wiese, M.; Wood, B.; Pachoud, A.; Korn, R.; Buehler, H.; Murray, P.; and Bai,
  L. 2021.
\newblock Multi-Asset Spot and Option Market Simulation.
\newblock arXiv:2112.06823.

\end{thebibliography}

\section{Disclosure}
Opinions and estimates constitute our judgement as of the date of this Material, are for informational purposes only and are subject to change without notice. It is not a research report and is not intended as such. Past performance is not indicative of future results. This Material is not the product of J.P. Morgan’s Research Department and therefore, has not been prepared in accordance with legal requirements to promote the independence of research, including but not limited to, the prohibition on the dealing ahead of the dissemination of investment research. This Material is not intended as research, a recommendation, advice, offer or solicitation for the purchase or sale of any financial product or service, or to be used in any way for evaluating the merits of participating in any transaction. Please consult your own advisors regarding legal, tax, accounting or any other aspects including suitability implications for your particular circumstances. J.P. Morgan disclaims any responsibility or liability whatsoever for the quality, accuracy or completeness of the information herein, and for any reliance on, or use of this material in any way. Important disclosures at: \href{www.jpmorgan.com/disclosures}{\texttt{www.jpmorgan.com/disclosures}}.

\appendix 
\onecolumn
\section{Additional numerical results}
\begin{figure*}[htp]
    \centering 
    \includegraphics[width=\textwidth]{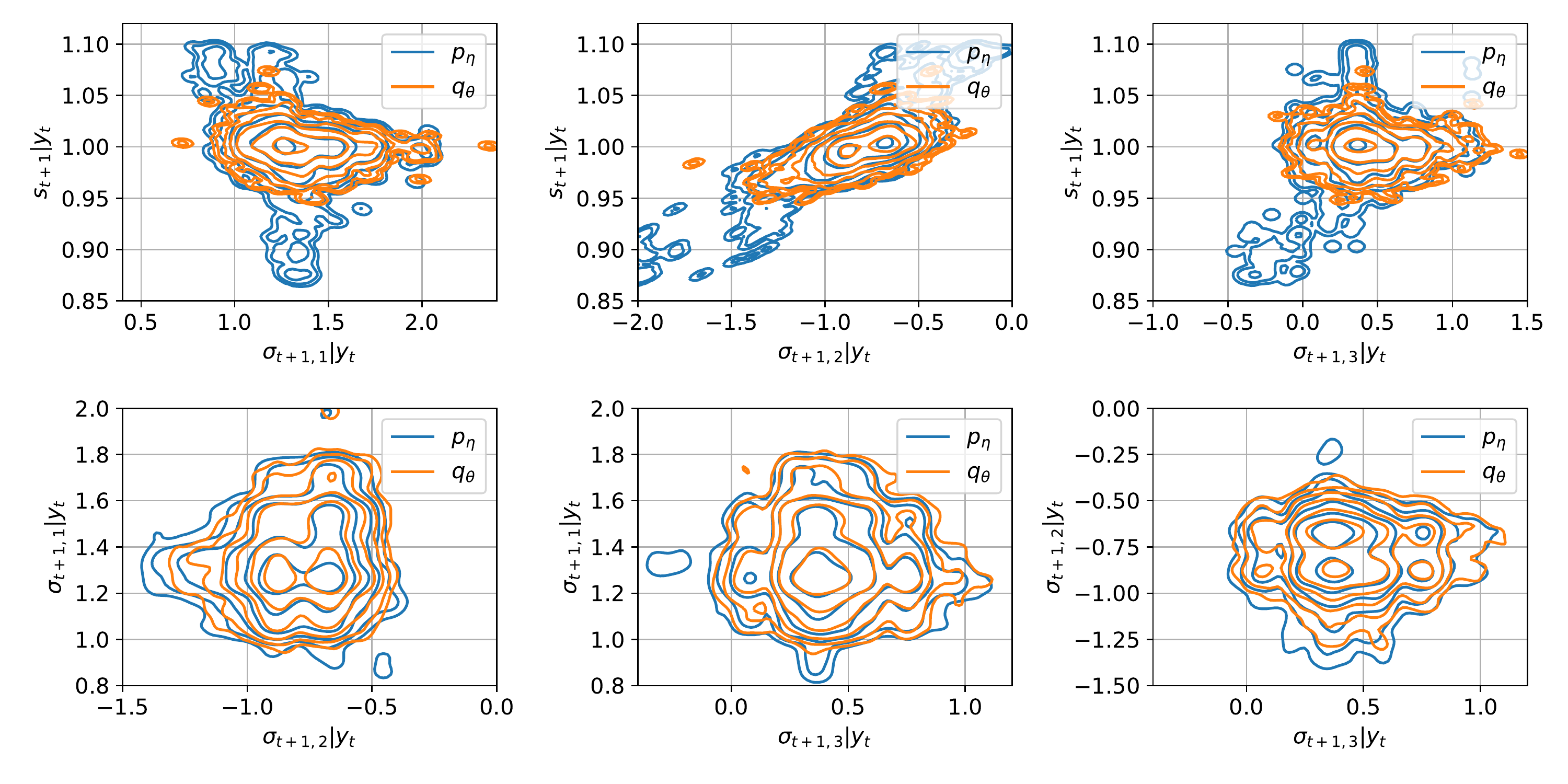}
    \caption{Kernel density estimate (contour plot) of the conditional spot and compressed state distribution $Y_{t+1} | Y_t = y_t$ sampled under the approximated real-world density $p_\eta$ and the approximated \emph{near} risk-neutral density $q_\theta$.}
    \label{fig:contour_plot}
\end{figure*}

\end{document}